\documentclass[a4paper,UKenglish,cleveref, autoref,thm-restate]{lipics-v2021}

\hideLIPIcs
\nolinenumbers

\title{Making Life More Confusing for Firefighters}

\author{Samuel D. Hand}{School of Computing Science, University of Glasgow, UK}{s.hand.1@research.gla.ac.uk}{0000-0001-8021-249X}{Supported by an EPSRC doctoral training account.}
\author{Jessica Enright}{School of Computing Science, University of Glasgow, UK}{jessica.enright@glasgow.ac.uk}{0000-0002-0266-3292}{Supported by EPSRC grant EP/T004878/1.}
\author{Kitty Meeks}{School of Computing Science, University of Glasgow, UK}{kitty.meeks@glasgow.ac.uk}{0000-0001-5299-3073}{Supported by EPSRC grant EP/T004878/1.}

\authorrunning{S.\,D. Hand, J. Enright, and K. Meeks}

\Copyright{Samuel D. Hand, Jessica Enright, and Kitty Meeks}

\ccsdesc[100]{Mathematics of computing~Graph algorithms}

\keywords{Temporal graphs, Spreading processes, Parameterised complexity}

\usepackage{standalone}
\usepackage{tikz}
\usetikzlibrary{graphs,graphs.standard,arrows.meta,positioning}

\DeclareMathOperator{\maxtime}{maxtime}
\DeclareMathOperator{\mintime}{mintime}
\DeclareMathOperator{\dist}{dist}

\newcommand{\decisionproblem}[3]{%
\vspace{0,1cm} \noindent \fbox{%
\begin{minipage}{0.96\textwidth}%
 \begin{tabular*}{\textwidth}{@{\extracolsep{\fill}}lr} \textsc{#1} & \\ \end{tabular*}%
  \vspace{1.2mm}%
\par%
{\bf{Input:}}  #2\\%
{\bf{Output:}} #3%
\end{minipage}} \vspace{0,3cm}%
}

\begin{document}

\maketitle

\begin{abstract}
It is well known that fighting a fire is a hard task. The \textsc{Firefighter} problem asks how to optimally deploy firefighters to defend the vertices of a graph from a fire. This problem is NP-Complete on all but a few classes of graphs. Thankfully, firefighters do not have to work alone, and are often aided by the efforts of good natured civillians who slow the spread of a fire by maintaining firebreaks when they are able. We will show that this help, although well-intentioned, unfortunately makes the optimal deployment of firefighters an even harder problem. To model this scenario we introduce the \textsc{Temporal Firefighter} problem, an extension of \textsc{Firefighter} to temporal graphs. We show that \textsc{Temporal Firefighter} is also NP-Complete, and remains so on all but one of the underlying classes of graphs on which \textsc{Firefighter} is known to have polynomial time solutions. This motivates us to explore making use of the temporal structure of the graph in our search for tractability, and we conclude by presenting an FPT algorithm for \textsc{Temporal Firefighter} with respect to the temporal graph parameter vertex-interval-membership-width.
\end{abstract}

\section{Introduction}
Imagine a fire breaks out on an island within an archipelago. The fire service always have one on duty firefighter, who is quickly deployed to protect one of the other islands. The islands within the archipelago are connected by bridges, which the fire now spreads along to all unprotected islands that it neighbours. By now, the fire service have called in another firefighter, who is again deployed to protect an island, and the process repeats. The question of determining how many islands can be saved from the fire in such a scenario is formalised by the \textsc{Firefighter} problem, which models the spread of the fire over the vertices of a graph \cite{hartnell1995firefighter}.

As noted by Fomin et al. firefighting is a tough job \cite{DBLP:conf/fun/FominHL12}. Specifically \textsc{Firefighter} is NP-Complete on arbitrary graphs, although progress has been made on identifying graph classes for which it can be solved in polynomial time \cite{DBLP:journals/dm/FinbowKMR07, DBLP:journals/tcs/FominHL16}. In particular these are: interval graphs, permutation graphs, $P_k$-free graphs for $k > 5$, split graphs, cographs, and graphs of maximum degree three providing the root is of degree two. Additionally, both the parameterised complexity and the approximability of the problem have been considered \cite{DBLP:journals/algorithmica/AnshelevichCHS12, DBLP:journals/jcss/BazganCCFFL14, DBLP:conf/isaac/CaiVY08, DBLP:journals/tcs/ChlebikovaC17}. For a more general review of known results about \textsc{Firefighter}, see the work by Finbow and MacGillivray \cite{DBLP:journals/ajc/FinbowM09}.

Thankfully for the firefighters, they don't have to do all the work themselves. When they can spare the time, islanders will help to maintain firebreaks at the bridges, thus delaying the spread of the fire. Whilst this help is of course much appreciated, it unfortunately makes the problem of choosing how to optimally the deploy firefighters all the more confusing, even when it is known ahead of time when the islanders will be available to maintain the firebreaks. In this paper we explore how to simplify this decision making, both by making use of the layout of the archipelago, and the availability of the islanders. In order to do this we introduce \textsc{Temporal Firefighter}, an extension of \textsc{Firefighter} to a variety of graph in which the edges of an underlying graph are assigned times at which they are active. We refer to this variant of graph as a temporal graph. Existing algorithmic work on temporal graphs has explored how they change the notions of both paths and connectivity \cite{DBLP:journals/corr/AxiotisF16,DBLP:conf/adhoc-now/BhadraF03,DBLP:journals/ijfcs/XuanFJ03,DBLP:journals/snam/HimmelMNS17,DBLP:journals/jcss/KempeKK02,DBLP:journals/tkde/WuCKHHW16}. For a survey of algorithmic work on temporal graphs see Michai \cite{DBLP:journals/im/Michail16}, and for a more multidisciplinary overview see the work by Holme and Saram{\"{a}}ki \cite{DBLP:journals/corr/abs-1108-1780}.

We begin in \autoref{sec:prelim} by giving a formal definition \textsc{Temporal Firefighter}, before exploring how extending \textsc{Firefighter} in this way affects the spread of the fire and the complexity of the associated decision problem. We find that for every class $\mathscr{C}$ of graphs for which \textsc{Firefighter} is NP-Complete, \textsc{Temporal Firefighter} is NP-Complete on the class of temporal graphs with the graphs of $\mathscr{C}$ underlying graphs. This motivates a search for tractable cases of \textsc{Temporal Firefighter} in two directions. Firstly, in \autoref{sec:restrict}, we explore its complexity when the underlying graph class is restricted, finding that it remains NP-Complete on all but one of the underlying graph classes for which \textsc{Firefighter} is tractable. More promisingly, in \autoref{sec:fpt}, we investigate restricting the temporal structure, and give an algorithm that is FPT with respect to the temporal graph parameter vertex-interval-membership-width. 

\section{Preliminaries}
\label{sec:prelim}
Formally, the \textsc{Firefighter} problem asks how many vertices it is possible to prevent from burning on a connected, undirected, loop-free, rooted graph in the following discrete time process:

\begin{enumerate}
    \item At time $t = 0$, the root is labeled as burning.
    \item At all times $t \geq 1$, a chosen vertex is labeled as defended, and the fire then spreads to all undefended vertices adjacent to the fire.
    \item This process ends once the fire can no longer spread.
\end{enumerate}

A vertex $v$ is valid to defend on timestep $i$ if and only if $v$ is not burning or already defended on timestep $i$. We refer to a sequence of such valid defences for \textsc{Firefighter} as a strategy.

\begin{definition}[A Strategy]
\label{def:strat}
A strategy is a sequence of vertices $v_1, v_2, ..., v_\ell$, such that each $v_i$ is a valid defence on timestep $i$.
\end{definition}

We say a vertex is saved if it is not burning once the process ends. The decision problem then asks how many vertices can be saved on a given graph:

\decisionproblem{Firefighter}{A rooted graph $(G, r)$ and an integer $k$.}{Does there exist a strategy that saves at least $k$ vertices on $G$ when the fire starts at vertex $r$?}

We now extend \textsc{Firefighter} to temporal graphs, using the definition of temporal graph first introduced by Kempe et al. \cite{DBLP:journals/jcss/KempeKK02}.

\begin{definition}[A Temporal Graph]
A pair $(G, \lambda)$ where $G$ is the underlying static graph $(V, E)$ and $\lambda : E \to 2^\mathbb{N}$ is the time-labeling function, assigning to each edge a set of timesteps at which it is active.
\end{definition}

The lifetime $\Lambda$ of a temporal graph refers to the final time at which any edge is active:

\begin{definition}[Lifetime]
The lifetime $\Lambda$ of a temporal graph $(G, \lambda)$ is the maximum time on any edge. $\Lambda = \max \{ \max \lambda(e) : e \in  E(G)\}$.
\end{definition}

Temporal graphs introduce a new notion of adjacency. We say that two vertices are temporally adjacent on a given timestep if there is an edge between them active on that timestep.

\begin{definition}[Temporal Adjacency]
Two adjancent vertices $v_1$ and $v_2$  in the temporal graph $(G, \lambda)$ are temporally adjacent at time $t$ if  $t \in \lambda(v_1, v_2)$.
\end{definition}

In \textsc{Temporal Firefighter}, just as in \textsc{Firefighter}, the fire begins burning at a root vertex $r$, and on each timestep a single vertex can be defended before the fire spreads. Unlike \textsc{Firefighter} the fire does not spread to all adjacent vertices, but only to vertices to which it is temporally adjacent.

We define a strategy for \textsc{Temporal Firefighter} exactly as in \autoref{def:strat} for \textsc{Firefighter}, and the decision problem is then defined analogously to \textsc{Firefighter}:

\decisionproblem{Temporal Firefighter}{A rooted temporal graph $((G, \lambda), r)$ and an integer $k$.}{Does there exist a strategy that saves at least $k$ vertices on $(G, \lambda)$ when the fire starts at vertex $r$?}

In some ways modifying \textsc{Firefighter} to take place on a temporal graph actually makes the job of the firefighters easier. Assigning times to the edges of a static graph only serves to limit the spread of the fire. In particular, the fire in \textsc{Temporal Firefighter} can only spread along temporally admissible paths, these being a subset of the paths in the underlying static graph.

\begin{definition}[Temporally Admissible]
A temporally admissible path on the temporal graph $(G, \lambda)$ is a path on $G$ with edges $e_1,...,e_\ell$, such that there is a strictly increasing sequence of times $t_1,...,t_\ell$ with $t_i \in \lambda(e_i)$ for every $i$.
\end{definition}

Furthermore, assigning times to the edges can only slow the rate at which the fire spreads down a path; the fire is still limited to spreading at a rate of at most one vertex per timestep along that path. We refer to the earliest time at which the fire can burn fully along the length of a path from the root as the arrival time of the path.

\begin{definition}[Arrival Time]
The arrival time of a temporally admissible path containing edges $e_1,...e_\ell$ on the temporal graph $(G, \lambda)$ is the minimum $t_\ell$ such that $t_\ell$ is the end of a strictly increasing sequence of times $t_1,...,t_\ell$ with $t_i \in \lambda(e_i)$ for every $i$.
\end{definition}

As a result, if the same defences are made, the fire cannot reach anywhere in \textsc{Temporal Firefighter} on a rooted temporal graph $((G, \lambda), r)$ that it would not be able to reach in \textsc{Firefighter} on the underlying static graph $(G, r)$. This gives us the following observation.

\begin{observation}
\label{ob:validstrat}
Any strategy $S = v_1, ..., v_\ell$ for \textsc{Firefighter} on a rooted graph $(G, r)$ is also a valid strategy for \textsc{Temporal Firefighter} on any rooted temporal graph $((G, \lambda), r)$. Furthermore any vertex saved by $S$ in \textsc{Firefighter} must also be saved by $S$ in \textsc{Temporal Firefighter}.
\end{observation}

However, the decision problem remains just as hard, as we can assign times in a rooted temporal graph $((G, \lambda), r)$ such that \textsc{Temporal Firefighter} simulates \textsc{Firefighter} for any rooted graph $(G, r)$. This is achieved by setting $\lambda(e) = \{1, ..., |V(G)|-1\}$ for every edge $e$. By time $|V(G)|-1$ every vertex would have been defended, so the process must be over. Thus, for the entirety of the time during which the fire can spread, every edge is active, just as in \textsc{Firefighter}. In this respect we can view \textsc{Firefighter} to be a special case of \textsc{Temporal Firefighter}.

Note that \textsc{Temporal Firefighter} is in NP, as a strategy acts as a certificate that can be checked in polynomial time by simulating \textsc{Temporal Firefighter}. We then have the following observation, as the above method for simulating \textsc{Firefighter} preserves the underlying graph class.

\begin{observation}
For every class $\mathscr{C}$ of graphs for which \textsc{Firefighter} is NP-Complete, \textsc{Temporal Firefighter} is NP-Complete on the class of temporal graphs with the graphs of $\mathscr{C}$ as the underlying graphs.
\end{observation}

\section{Restricting the Underlying Graph}
\label{sec:restrict}
As we have seen that \textsc{Temporal Firefighter} is NP-Complete on any class of temporal graphs $\{(G, \lambda) : G \in \mathscr{C}\}$ where $\mathscr{C}$ is a class of graphs for which \textsc{Firefighter} is NP-Complete, we now determine its complexity on underlying graph classes for which \textsc{Firefighter} is known to be solvable in polynomial time. These are: interval graphs, permutation graphs, $P_k$-free graphs for $k > 5$, split graphs, cographs, and graphs of maximum degree three providing the root is of degree two\cite{DBLP:journals/dm/FinbowKMR07,DBLP:journals/tcs/FominHL16}. We prove that \textsc{Temporal Firefighter} is NP-Complete for all of these classes except the last, for which we find there is a polytime solution. Additionally we establish that it is NP-Complete for AT-free graphs, a class for which the complexity of \textsc{Firefighter} has not been determined. This lack of tractable cases of \textsc{Temporal Firefighter} when restricting the underlying graph motivates restricting the temporal structure, which we explore in \autoref{sec:fpt}, finding this to be a fruitful route to tractability.

\begin{figure}
     \centering
     \begin{tikzpicture}
    \begin{scope}[shift={(-5,0)}]
    \graph[circular placement, radius=4cm,
           empty nodes, nodes={circle,draw}] {
        a [as=$r$];
        b;
        c;
        d;
        e;
        f;
        e -- b -- d -- c -- a -- f;
        a -- e;
    };
    \end{scope}
    \draw [double distance = 5, -Implies] (-0.5, 0) -- (0.5, 0);
    \begin{scope} [shift={(5,0)}]
    \graph[circular placement, radius=4cm,
           empty nodes, nodes={circle,draw}] {
        a [as=$r$];
        b;
        c;
        d;
        e;
        f;
        e -- b -- d -- c -- a -- f;
        a -- e;
        { [edges=dashed]
            a -- b -- c -- e -- f;
            d -- e;
            b -- f;
            d -- f;
        }
    };
    \end{scope}
    \begin{scope} [shift = {(9.5,0)}, local bounding box = legend, node distance = 0.1cm and 0.5cm ]
        \node [anchor = south] (existing1) at (0, 0.05) {$\lambda($};
        \node (existing2) [right = of existing1] {$)=\{1..|V(G)|-1\}$};
        \draw (existing1) -- (existing2);
        \node (new1) [below = of existing1] {$\lambda($};
        \node (new2) [right = of new1] {$)=\{|V(G)|-1\}$};
        \draw [dashed] (new1) -- (new2);
    \end{scope}
    \draw (legend.south west) rectangle (legend.north east);
\end{tikzpicture}
     \caption{An example of the reduction for \textsc{Temporal Firefighter} on cliques.}
     \label{fig:cliquereduction}
\end{figure}
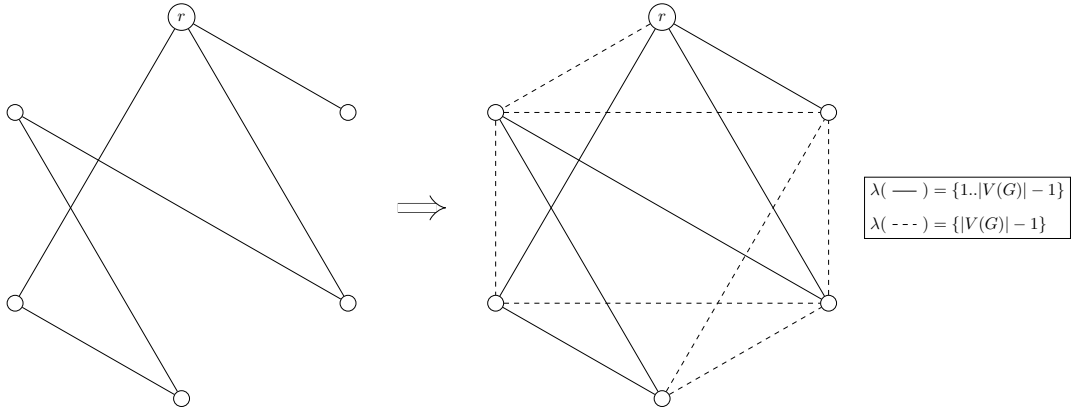

All of these hardness results follow from the fact that \textsc{Temporal Firefighter} is hard when the underlying graph is a clique. This can be shown by reduction from \textsc{Firefighter} by assigning times to the edges in a static graph $G$ so that they will be active at all times up until $|V(G)|-1$, at which point the fire can certainly no longer spread. We then add further edges to make the graph a clique, and have them only active from time $|V(G)|-1$ onwards such that they will not affect the spread of the fire. \textsc{Temporal Firefighter} on such a clique will then simulate \textsc{Firefighter} on $G$. A sketch of this construction can be seen in \autoref{fig:cliquereduction}.

We in fact prove the stronger result that \textsc{Temporal Firefighter} is hard on cliques of $n$ vertices with lifetime of less than $n^\frac{1}{c}$ for any positive integer constant $c$. This reduction operates by adding $n^c-n$ vertices to a static graph, and assigning times in such a way that they will all burn immediately, without affecting the spread of the fire over the existing graph. All defences then take place on a clique constructed in the same manner as that described above.

We first show that we can add edges to a temporal graph in such a way that they do not affect which vertices can be saved.

\begin{lemma}
\label{ob:addnonedgessave}
Suppose there is a strategy $S = v_1,...,v_\ell$ for \textsc{Temporal Firefighter} on the rooted temporal graph $(((V, E), \lambda), r)$ that saves $k$ vertices. Let $F$ be any set of additional edges not in $E$, and $\lambda' : E \cup F \to \mathbb{N}$ be a labelling function with $\lambda'{\big|}_E = \lambda$ and $\min(\lambda'(f)) \geq |V|-1$ for all $f \in F$. Let $S'$ be the strategy consisting of all the defences in $S$ followed by defending every remaining undefended vertex in an arbitrary order. $S'$ will then save $k$ vertices in \textsc{Temporal Firefighter} on $(((V, E \cup F), \lambda'), r)$.
\end{lemma}
\begin{proof}
For each timestep $t \leq \ell$ consider any vertex $v$ that does not burn by the end of timestep $t$ when the first $t$ defences from $S$ are played on $(((V, E), \lambda), r)$. We will show by induction on $t$ that this vertex does not burn when the first $t$ defences from $S$ are played on $(((V, E \cup F), \lambda'), r)$. See that in particular this allows us to inductively assume that the defences are valid on $(((V, E \cup F), \lambda'), r)$, as in particular the inductive hypothesis will imply that for any $t' \leq t$, a defence $v_t'$ will not burn before the end of timestep $t' - 1$ on $(((V, E \cup F), \lambda'), r)$.

If $t = 0$ then the only vertex to have burnt in both graphs is $r$. Otherwise, consider all the paths from $r$ to $v$ in $(((V, E), \lambda), r)$. As $v$ does not burn, each of these paths either contains a defended vertex or has an arrival time greater than $t$. Now consider the paths from $r$ to $v$ in $(((V, E \cup F), \lambda'), r)$. For each of these paths, either it is one of the aforementioned paths from $(((V, E), \lambda), r)$, or contains an edge from $F$ and thus has an arrival time of at least $|V|-1 > t$. In either case, the fire cannot have burnt along the path to $v$, and thus $v$ does not burn.

Finally note that there is time to make the extra defences in $S'$ before the fire spreads down the additional edges in $F$, as all vertices in the graph must be defended by the time these edges are active.
\end{proof}

We are now ready to give the reduction. This result allows us to determine that \textsc{Temporal Firefighter} is NP-Complete on the class of temporal graphs $\{((G, \lambda), r) : (G, r) \in \mathscr{C}\}$ for all but one of the classes $\mathscr{C}$ for which it is known that \textsc{Firefighter} has a polynomial time solution.

\begin{figure}
     \centering
     \begin{tikzpicture}
    \begin{scope}[shift={(-5,0)}]
    \graph[circular placement, radius=4cm,
           empty nodes, nodes={circle,draw}] {
        a [as=$r$];
        b;
        c;
        d;
        e;
        f;
        e -- b -- d -- c -- a -- f;
        a -- e;
    };
    \end{scope}
    \draw [double distance = 5, -Implies] (-0.5, 0) -- (0.5, 0);
    \begin{scope} [shift={(5,0)}]
    \graph[circular placement, radius=4cm,
           empty nodes, nodes={circle,draw}] {
        a [as=$r$];
        b;
        c;
        d;
        e;
        f;
        w[as=$\{W\}$, ,yshift=2cm,xshift=2cm] -- a;
        e -- b -- d -- c -- a -- f;
        a -- e;
        { [edges=dashed]
            a -- b -- c -- e -- f;
            d -- e;
            b -- f;
            d -- f;
            w -- b;
            w -- c;
            w -- d;
            w -- e;
            w -- f;
        }
    };
    \end{scope}
    \begin{scope} [shift = {(9.5,1)}, local bounding box = legend, node distance = 0.1cm and 0.5cm ]
        \node [anchor = south] (existing1) at (0, 0.05) {$\lambda($};
        \node (existing2) [right = of existing1] {$)=\{1..|V(G)|-1\}$};
        \draw (existing1) -- (existing2);
        \node (new1) [below = of existing1] {$\lambda($};
        \node (new2) [right = of new1] {$)=\{|V(G)|-1\}$};
        \draw [dashed] (new1) -- (new2);
    \end{scope}
    \draw (legend.south west) rectangle (legend.north east);
\end{tikzpicture}
     \caption{An example of the reduction for \textsc{Temporal Firefighter} on cliques with bounded lifetime. The vertex marked $\{W\}$ represents the set $W$ containing $|V(G)|^c-|V(G)|$ vertices.}
     \label{fig:boundedcliquereduction}
\end{figure}
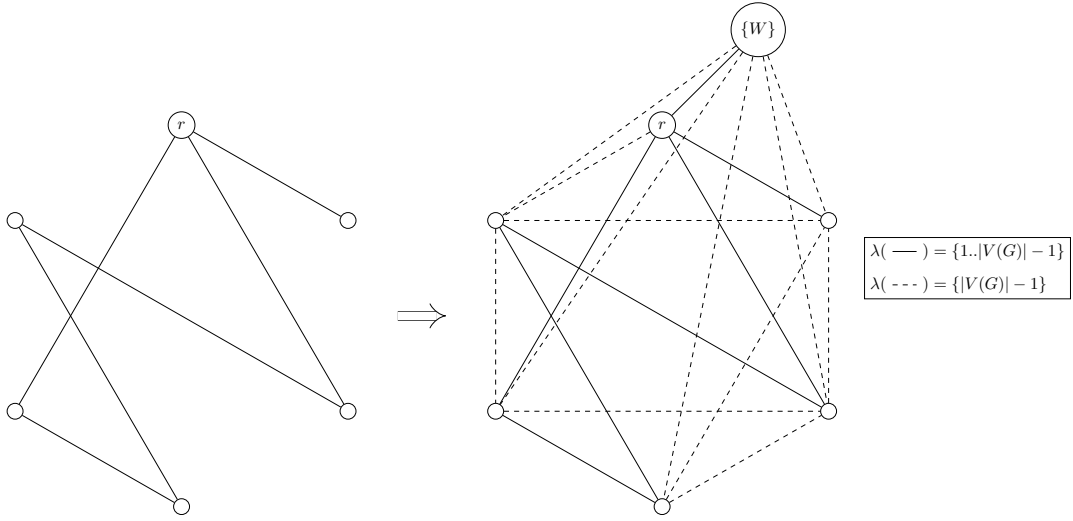

\begin{theorem}
\label{thm:hardcliquebounded}
For any constant $c \in \mathbb{N}$, \textsc{Temporal Firefighter} is NP-Complete when restricted to temporal graphs whose underlying graph is a clique and whose lifetime is at most $n^\frac{1}{c}$ where $n$ is the number of vertices in the graph.
\end{theorem}
\begin{proof}
We give a reduction from \textsc{Firefighter}; given an instance $((G, r), k)$ of \textsc{Firefighter} and a constant $c$ we construct an instance $(((G', \lambda), r), k)$ of \textsc{Temporal Firefighter} which is a yes-instance if and only if $(G, r), k)$ is a yes-instance of \textsc{Firefighter}.

Letting $\ell = |V(G)|$, we now construct an instance $(((G', \lambda), r), k)$ of \textsc{Temporal Firefighter} with $\ell^c$ vertices as follows. Let $W$ be a set of $\ell^c-\ell$ vertices not in $G$. Then let $G'$ be the graph $(V', E')$, with $V' = V(G) \cup W$, and $E'$ containing edges connecting every pair of distinct vertices in $V'$, making the graph a clique, as shown in \autoref{fig:boundedcliquereduction}. We then define $\lambda$ as follows:

\[
\lambda(e) =
    \begin{cases}
        \{1,2,...,\ell-2\} & \text{if } e \in E(G) \\
        & \text{or } e = rv \text{ and } v \in W \\
        \{\ell-1\} & \text{otherwise.}
    \end{cases}
\]

Note that the lifetime of this instance is $\ell-1$, which is less than $|V(G')|^\frac{1}{c} = \ell$, as required.

We now show that if $((G, r), k)$ is a yes-instance of \textsc{Firefighter} then $(((G', \lambda), r), k)$ is a yes-instance of \textsc{Temporal Firefighter}. If $((G, r), k)$ is a yes-instance then there is a strategy $S$ for \textsc{Firefighter} on $(G, r)$ that saves at least $k$ vertices. We claim that we can save at least $k$ vertices in \textsc{Temporal Firefighter} on $((G', \lambda), r)$ by first playing the defences from $S$, and then defending in arbitrary order the remaining unburnt vertices in $V(G') \setminus W$.

If we play the defences from $S$, then every vertex in $W$ burns on the first timestep, and we are left with only the vertices in $V(G)$ to defend. Any path from a vertex in $W$ to a vertex in $V(G)$ other than those that go via the root has an arrival time of at least $\ell-1$, and we have at most $\ell-2$ vertices left to defend after $W$ burns, so the process must have ended before the fire spreads from $W$ into $V(G)$. We can then consider only the spread of the fire over the subgraph of $((G', \lambda), r)$ induced by $V(G)$, and this is the temporal graph $((V(G), E(G) \cup F), \lambda')$ where $F$ contains edges connecting every pair of vertices in $V(G)$ not connected by edges in $E(G)$, and $\lambda'$ is defined as follows. $$\lambda'(e) = \begin{cases}\{1,2,...,\ell-2\} & \text{if } e \in E(G) \\\{\ell-1\} & \text{otherwise}\end{cases}$$ We know from \autoref{ob:validstrat} that the strategy $S$ will save at least $k$ vertices in any temporal graph with $G$ as the underlying graph, and then from \autoref{ob:addnonedgessave} we know that it is possible to save at least $k$ vertices on $((V(G), E(G) \cup F), \lambda')$, and thus $(((G', \lambda), r), k)$ is a yes-instance.

To show the converse, we first argue that if it is possible to save $k$ vertices in \textsc{Temporal Firefighter} on $((G', \lambda), r)$ then in particular it is possible to do this without defending any vertices in $W$. It is only possible to defend a vertex in $W$ on the first timestep, as every undefended vertex in $W$ will burn on timestep 1. As $G$ is connected, there must be at least one vertex $v$ in $V(G)$ on $((G', \lambda), r)$ that is connected to $r$ by an edge active at times $\{1,2,...,\ell-2\}$, and so $v$ will burn on the first timestep if a vertex in $W$ is defended. Thus, defending $v$ instead of a vertex in $W$ on the first timestep saves at least as many vertices.

Next we observe that in any strategy that does not defend a vertex in $W$, the fire stops spreading by timestep $\ell-1$, as every vertex in $W$ burns instantly on timestep 1, and by timestep $\ell-1$ it must be the case that every vertex in $V(G)$ is burnt or defended.

It follows that if $(((G', \lambda), r), k)$ is a yes-instance then there is a strategy for \textsc{Temporal Firefighter} on $((G', \lambda), r)$ that saves at least $k$ vertices, and does not defend any vertices in $W$.

We now see that this same strategy is valid for \textsc{Firefighter} on $(G, r)$. Firstly, it does not defend any vertices in $W$. Secondly, if we consider any defence $v_i$ in the strategy, then we can see that the paths from the root $r$ to $v_i$ in $(G, r)$ are all also present in $((G', \lambda), r)$ and have arrival times equal to their length. If $v_i$ does not burn in \textsc{Temporal Firefighter} on $((G', \lambda), r)$ then every such path is either defended, or has length  greater than $i$. Thus, if we inductively assume the first $i-1$ defences from $S$ are valid on $(G, r)$, we can see that when these defences are played $v_i$ cannot burn in \textsc{Firefighter} on $(G, r)$, as every undefended path from $r$ to $v_i$ must have length greater than $i$.

Furthermore if a vertex $v$ does not burn in \textsc{Temporal Firefighter} on $((G', \lambda), r)$ then it must not burn in \textsc{Firefighter} on $(G, r)$, as the temporally admissible paths between $r$ and $v$ in $(G', \lambda)$ are a superset of the paths between $r$ and $v$ in $G$. If $v$ does not burn in \textsc{Temporal Firefighter} then each of these paths either contains a defended vertex, or has an arrival time of $\ell-1$, and is not present in $G$. Therefore it is possible to save at least $k$ vertices on $(G, r)$, and $((G, r), k)$ is also a yes-instance.
\end{proof}

As a result we can deduce that \textsc{Temporal Firefighter} is NP-Complete on several clique containing classes for which \textsc{Firefighter} is in P. For the same reason, we can determine that \textsc{Temporal Firefighter} is NP-Complete on AT-free graphs, a class for which the complexity of \textsc{Firefighter} is still an open problem.

\begin{corollary}
\textsc{Temporal Firefighter} is NP-Complete on split graphs, unit interval graphs, cographs, $P_k$-free graphs for $k > 2$, and AT-free graphs.
\end{corollary}

We have seen, \textsc{Temporal Firefighter} is hard on several graph classes for which \textsc{Firefighter} is easy. However there is one non-trivial class for which both \textsc{Firefighter} and \textsc{Temporal Firefighter} are easy, that being the class of graphs of maximum degree three, with a root of degree at most two.

A proof that \textsc{Firefighter} is easy on this class is given by Finbow et al. \cite{DBLP:journals/dm/FinbowKMR07}. This proof works due to the fact that it is always optimal to restrict the fire to spreading down only one path on such a graph. An algorithm need only find the shortest path at which the fire can be contained at the end, and then defend accordingly. Exactly the same can be done for for \textsc{Temporal Firefighter}, the only difference being in calculating where the fire can be contained -- sometimes the active times of the edges allow the fire to be contained at a vertex in \textsc{Temporal Firefighter} where it could not be contained in \textsc{Firefighter}.

We present a strategy $S$ for \textsc{Temporal Firefighter} on a temporal graph $((G, \lambda), r)$ of maximum degree three where $r$ is of degree two, and show that this strategy is optimal and computable in polynomial time. This strategy and proof only requires slight modifications from that given by Finbow et al. for \textsc{Firefighter} \cite{DBLP:journals/dm/FinbowKMR07}.

Throughout, for any two vertices $v$ and $u$ in a temporal graph $(G, \lambda)$ let $\dist(v, u)$ be the number of edges on the shortest path between $v$ and $u$ in the underlying graph $G$.

After defining the strategy $S$ we show that no strategy that does not always defend next to the fire can outperform $S$, and thus that there always exists an optimal strategy which only defends next to the fire. Such a strategy, due to the degree restriction, limits the fire to spreading along a single path. An optimal strategy then finds the shortest of such paths to a vertex at which the spread of the fire can be stopped. We observe that this can be done at any vertex $u$ where there are one or less incident edges not on the path and active on timestep $\dist(r, u) + 1$. Stated otherwise, as soon as the fire reaches a vertex at which the temporal nature of the graph delays its spread, it is possible to contain it, and in an optimal strategy the fire will spread along the path to such a vertex at a rate of one vertex per timestep, just as in \textsc{Firefighter}. We then show that strategy $S$ is exactly this strategy, and then that the number of vertices saved by such an optimal strategy can be computed in polynomial time, as required.

We begin by defining three sets that will be used in the strategy: $V_0$, $V_1$, and $V_c$.

$V_0$ and $V_1$ are the sets of all vertices $u$ that at time $\dist(r, u)+1$ are temporally adjacent, respectively, to $0$ and $1$ vertices not on the shortest underlying path between $r$ and $u$. $V_c$ is the set of all vertices that lie on a cycle and are not in $V_0$ or $V_1$. Additionally for any vertex $u$, $C(u)$ denotes the length of the shortest cycle containing $u$.

Strategy $S$ operates by first finding a vertex $u \in V_0 \cup V_1 \cup V_c$ that minimizes the function $f(u)$, defined below.

\[
f(u) =
    \begin{cases}
        \dist(r, u) + 1 & \text{if } u \in V_0 \cup V_1 \\
        \dist(r, u) + C(u) - 1 & \text{if } u \in V_c
    \end{cases}
\]
If $u \in V_0 \cup V_1$, then let $P$ be the shortest path from $r$ to $u$ on the underlying graph $G$. As $u$ minimizes $f$, this path will always be temporally admissible and have an arrival time equal to its length.

If the path was not temporally admissible, or did not have an arrival time equal to its length, then there would be a vertex $v$ on the path and closer to $r$ than $u$ which would not be temporally adjacent to the next vertex on the path. Thus $v \in V_0 \cup V_1$, and $f(v) < f(u)$.

The strategy is then to always defend the vertex adjacent to the fire that does not lie on $P$, up until turn $f(u)$. On turn $f(u)$ a non-burning neighbour of $u$ should be defended, prioritising a temporally adjacent neighbour if one exists. If there is a further non-burning neighbour of $u$, this should be defended on turn $f(u)+1$. Once the fire stops spreading, the burnt vertices will be all those on $P$, meaning that in total $f(u)$ vertices will be burnt.

Otherwise, if $u \in V_c$, then let $C$ be the shortest cycle containing $u$, and $P$ the shortest path from $r$ to $u$.

Note that as $f(u)$ is minimal, it must be the case that $P$ is either of length 0, or does not contain any edges of $C$. If it did, then there would be a vertex $v$ on $P$ and $C$ with a neighbour on $P$ but not on $C$. As $v$ lies on the shortest path between $r$ and $u$, it is necessarily closer to $r$ than $u$. Additionally, as $v$ lies on the same cycle as $u$, we would have that $C(v) \leq C(u)$, and thus $f(u) \leq f(u)$ which is impossible.

The strategy is then to always defend the vertex adjacent to the fire that does not lie on $P$, up until turn $\dist(r, u)+1$. On turn $\dist(r, u)+1$, one of the two non-burning vertices on $C$ adjacent to the fire should be defended. On each following turn, the vertex adjacent to the fire but not on $C$ should be defended. Once the fire stops spreading, the burnt vertices will be all those on $P$ and all those on $C$ except one, meaning once again $f(u)$ vertices are burnt in total.

We now show that there is an optimal strategy that always defends next to the fire. The argument here is equivalent to that for \textsc{Firefighter} \cite{DBLP:journals/dm/FinbowKMR07}, but uses comparison to our newly defined strategy $S$ to show optimality.

\begin{lemma}
\label{lem:adjfire}
Given a rooted temporal graph $((G, \lambda), r)$ of maximum degree $3$ and with a root of degree $2$, there is an optimum strategy that always defends next to the fire.
\end{lemma}
\begin{proof}
Assume there is some counterexample minimal in number of vertices, that is a graph $((G, \lambda), r)$ with no optimal strategy that always defends next to the fire. Let $x_1$ and $x_2$ be the two neighbours of $r$. If there is an optimal strategy in which the first vertex defended is a neighbour of $r$, say $x_1$ without loss of generality, then $((G - \{r, x_1\}, \lambda), x_2)$ is a smaller counterexample -- a contradiction.

Let $T$ be some optimal strategy for \textsc{Temporal Firefighter} on $((G, \lambda), r)$, and let $u$ be the closest vertex to $r$ defended in $T$. This cannot be a neighbour of $r$, and thus $\dist(r, u) \geq 2$.

If no neighbours of $u$ are burning once the fire stops spreading, then there are no temporally admissible paths between $r$ and $u$. In this case $T$ wastes a defence defending $u$, a vertex that will never burn, and is thus non-optimal, a contradiction.

If only one neighbour of $u$ is burning once the fire stops spreading, then defending this neighbour instead of $u$ saves one more vertex, and thus $T$ is once again non-optimal.

If once the fire stops spreading two neighbours of $u$ are burning, then $u$ lies on a cycle that is completely burnt except for $u$. In this case we argue that the strategy $S$ must save at least as many vertices. Strategy $S$ finds a vertex $v$ that minimizes $f(v)$, and always causes $f(v)$ vertices to burn, saving the rest. If $T$ performs better than $S$, then there is a vertex $w \in V_c$ lying on the same cycle as $u$ where the entire path between $r$ and $w$ has burnt, as well as the entire cycle except for $u$, thus meaning that $f(w) < f(v)$ vertices burn. This is impossible -- $f(v)$ is a minimum. As $S$ always defends next to the fire, its optimality contradicts the assumption.
\end{proof}

We now show that strategy $S$ is an optimal strategy, and thus that \textsc{Temporal Firefighter} is in P for temporal graphs of maximum degree $3$ with roots of degree $2$.

\begin{theorem}
\label{thm:threep}
\textsc{Temporal Firefighter} can be solved in polynomial time on a rooted temporal graph $((G, \lambda), r)$ of maximum degree $3$ with a root of degree at most $2$.
\end{theorem}
\begin{proof}
First we note that if strategy $S$ is played on the graph $((G, \lambda), r)$ then $\min\{ f(u) \mid u \in V_0 \cup V_1 \cup V_c \}$ vertices will burn, and furthermore this value can be computed in polynomial time. We now show that strategy $S$ is optimal, and thus \textsc{Temporal Firefighter} is in P for temporal graphs of maximum degree $3$ with roots of degree $2$.

By Lemma \ref{lem:adjfire} there is an optimal strategy $T$ in which each vertex defended is next to the fire, thus restricting the fire to spreading down a single path. Let $w$ be the final vertex to burn, at the end of this path.

Due to the degree restriction there are at most two vertices adjacent to $w$ that do not lie on the path from $r$ down which the fire burnt to reach $w$. There are two ways in which the fire can stop spreading at $w$. In the first case both of these vertices are defended after the fire reaches $w$, and in the second at least one of these vertices has already been defended before the fire reaches $w$, and any undefended neighbours are defended afterwards.

In the first case, we must have that $w \in V_0 \cup V_1$, as there must have been time to make these defences after the fire reached $w$. Furthermore at least $\dist(r, w) + 1$ vertices must have burnt, and strategy $S$ performs at least as well, and is therefore optimal.

Otherwise, in the second case, $w$ must lie on a cycle that is fully burnt except for one vertex, as the already defended vertex must be adjacent to some burning vertex, and therefore adjacent to a vertex that lies on the burnt path from $r$ to $w$, as these are the only vertices that burn. Let $v$ be the first vertex on $C$ to have burnt. A path from $r$ to $v$ must be fully burnt, as is all of $C$ except for one vertex. Thus $f(v)$ vertices have burnt in total. As strategy $S$ finds a vertex $u$ that minimises $f(u)$, and allows $f(u)$ vertices to burn, it must be the case that $f(u) = f(v)$, and thus strategy $S$ is optimal. 
\end{proof}

\section{Restricting the Temporal Structure}
\label{sec:fpt}
As we have seen, our firefighters are going to have a hard time deciding on an optimal deployment strategy regardless of the layout of the archipelago. Our analysis of the complexity of \textsc{Temporal Firefighter} when restricting the underlying graph class shows that for most known graph classes $\mathscr{C}$ where \textsc{Firefighter} is polytime solvable, \textsc{Temporal Firefighter} is NP-Complete on the class of temporal graphs $\{(G, \lambda) : G \in \mathscr{C}\}$. This naturally leads us to consider whether the firefighters might be able to make use of some structure in the availability of the islanders, rather than the static layout of the islands. We now discuss the tractability of \textsc{Temporal Firefighter} when restricting the temporal structure of the graph.

We show that \textsc{Temporal Firefighter} is fixed parameter tractable when parameterised by vertex-interval-membership-width. Intuitively, bounding this parameter limits how active the graph can be on any given timestep.

Vertex-interval-membership-width, along with the vertex interval membership sequence, was defined by Bumpus and Meeks \cite{bumpus2021edge}. Begin by letting $\mintime(v)$ denote the minimum timestep upon which an incident edge of $v$ is active for all vertices $v$. Define $\maxtime$ equivalently for the maximum timestep.

\begin{definition}[Vertex Interval Membership Width]
The vertex interval membership sequence of a temporal graph $(G, \lambda)$ is the sequence $(F_t)_{t\in[\Lambda]}$ of vertex-subsets of $G$ where $F_t = \{v \in V(G) : \mintime(v) \leq t \leq \maxtime(v) \}$ and $\Lambda$ is the lifetime of $(G, \lambda)$.

The vertex-interval-membership-width of a temporal graph $(G, \lambda)$ is then the integer $\omega = \max_{t \in [\Lambda]}|F_t|$.
\end{definition}

Note that a vertex $v$ can only be in a sequence of consecutive members of the interval membership sequence. That is, it is impossible for there to be a vertex $v$ and times $s$, $t$ and $u$ such that $s < t < u$ where $v$ is in $F_s$ and $F_u$ but not in $F_t$. Furthermore, the vertex interval membership sequence of a graph can be computed in polynomial time \cite{bumpus2021edge}, and thus so can the vertex-interval-membership-width.

We find that \textsc{Temporal Firefighter} is FPT when parameterised by vertex-interval-membership-width. To simplify our analysis when showing this, we actually use the related problem \textsc{Temporal Firefighter Reserve}.

\textsc{Temporal Firefighter Reserve} is the temporal extension of the \textsc{Firefighter Reserve} problem described by Fomin et al. \cite{DBLP:journals/tcs/FominHL16}. In \textsc{Temporal Firefighter Reserve}, it is not required to make a defensive move every timestep. Rather, each timestep a budget is incremented by $1$, and it is then possible to defend any number of vertices less than or equal to the budget simultaneously, subtracting from the budget appropriately.

Just as in the static case, allowing the defence to build up a reserve in this manner does not affect the number of vertices than can be saved. In fact, the proof works identically to that for the static case as given by Fomin et al. \cite{DBLP:journals/tcs/FominHL16}.

\begin{lemma}
\label{lem:savereserve}
It is possible to save at least $k$ vertices in \textsc{Temporal Firefighter Reserve} on $((G, \lambda), r)$ if and only if it is possible to save at least $k$ vertices in \textsc{Temporal Firefighter Reserve} on $((G, \lambda), r)$.
\end{lemma}
\begin{proof}
Given a temporal graph $((G, \lambda), r)$, assume there is a strategy for \textsc{Temporal Firefighter Reserve} that saves $k$ vertices. Any strategy for \textsc{Temporal Firefighter} is a valid strategy for \textsc{Temporal Firefighter Reserve}, and thus it is also possible to save $k$ vertices in \textsc{Temporal Firefighter Reserve} by playing the same strategy.

Now assume that there is a strategy that saves $k$ vertices in \textsc{Temporal Firefighter Reserve} on $((G, \lambda), r)$. We can transform this strategy into a valid strategy for \textsc{Temporal Firefighter} that saves the same number of vertices as follows: if at any timestep $t$ the strategy defends $d > 1$ vertices, there must have been $d - 1$ timesteps at some point prior to this where no defences were made. By making exactly one of these $d$ defences on each of these $d-1$ prior timesteps, and timestep $t$ itself, we produce a valid strategy for \textsc{Temporal Firefighter}. Modifying the strategy in this manner creates a valid strategy, as if defending vertex $v$ is valid on timestep $t$, it must also be valid at any timestep less than $t$. Finally, as the exact same defences occur, only at an earlier time, the modified strategy must also save at least $k$ vertices.
\end{proof}

Additionally, we note that in \textsc{Temporal Firefighter Reserve} there is always an optimal strategy that only defends temporally adjacent to the fire, as any defence can be delayed until the turn upon which the defended vertex would burn. More generally, there is always an optimal strategy which defends only vertices at time $i$ if they have an incident edge active at time $i$. From now on when we refer to strategies for \textsc{Temporal Firefighter Reserve}, we assume that they all have this property.

We now give an algorithm for \textsc{Temporal Firefighter Reserve} that iterates over the vertex interval membership sequence of the input graph, and show that it is an FPT-algorithm with respect to vertex-interval-membership-width.

The algorithm takes as input a rooted temporal graph $((G, \lambda), r)$, and an integer $k$, and determines if it is possible to save $k$ vertices in temporal firefighter reserve played on the graph.
For any edge set $A$, let $V(A)$ be the set of vertices with an incident edge in $A$. The algorithm then operates by recursively computing a sequence of sets $L_i \in \mathcal{P}(F_i) \times \mathcal{P}(F_i) \times \{1, 2, ..., \Lambda\} \times \{1, 2, ..., n\}$ for each $F_i$ in the vertex interval membership sequence of the input graph.

An element of $L_i$ is a 4-tuple $(D, B, g, c)$ where $D$ is a set of defended vertices in $F_i$, $B$ is a set of burnt vertices in $F_i$, $g$ is the budget that will be available on timestep $i+1$, and $c$ is the total count of vertices that have burnt at time $i$. 

To determine the spread of the fire it is only necessary to keep track of the vertices that have burnt or been defended in $F_i$, as if a vertex is not in $F_i$ all its incident edges must either only be active before time $i$, or after time $i$. If the former is the case, then the fire cannot spread from or to it after time $i$, meaning that whether it is burning or defended does not affect the spread of the fire after this point. If the latter is the case then the vertex cannot be burning, as the fire cannot have reached it yet, and as we only defend vertices with incident edges active at time $i$, it cannot be defended either.

Additionally, it is possible to compute these defended and burning sets recursively from only a previous entry in the sequence, as a vertex $v$ can only be in a sequence of consecutive $F_i$s.
 
The problem can then be answered by checking if there is any entry $(D, B, g, c) \in L_\Lambda$ where $\Lambda$ is the lifetime of the graph, such that $|V(G)|-c \geq k$.

We recursively compute the sequence $L_i$, beginning by initialising $L_0 = (\emptyset, \{r\}, 1, 1)$. We let $E_i$ be the set of edges active at time $i$, and $N_i(S)$ the set of all vertices temporally adjacent at time $i$ to the vertices in $S$, for any set $S \subseteq V(G)$. Note that $V(E_i) \subseteq F_i$, as if a vertex has an incident edge active at time $i$, it is certainly in $F_i$.

We then require that, for any set $A \subseteq V(E_i) \setminus (B \cup D)$ containing vertices to be defended on timestep $i$, $(D', B', g-|A|+1, c') \in \mathcal{P}(F_i) \times \mathcal{P}(F_i) \times \{1, 2, ..., \Lambda\} \times \{1, 2, ..., n\}$ is in $L_i$ if and only if there is a tuple $(D, B, g, c)$ in $L_{i-1}$, such that:

\begin{enumerate}[(1)]
    \item $D' = (D \cap F_i) \cup A$
    \item $B' = (B \cap F_i) \cup N_i(B) \setminus D'$
    \item $g-|A|+1 > 0$
    \item $c' = c + |N_i(B) \setminus (B \cup D')|$
\end{enumerate}

That is, given sets of burning and defended vertices in $F_{i-1}$, we consider all the possible defences on vertices with incident edges active at time $i$, and create sets of burning and defended vertices in $F_i$ appropriately.
\par
Condition \textbf{1} ensures that the defended set contains only vertices with incident edges in $F_i$, and contains the set of new defences $A$.
\par
Condition \textbf{2} specifies that the burning set contains only vertices with incident edges in $F_i$, and that all non-defended vertices temporally adjacent to the fire burn.
\par
Condition \textbf{3} ensures that the budget is correct. The budget available on timestep $i+2$ will be $g-d+1$ if on timestep $i+1$ the budget is $g$, and $d$ vertices are to be defended, as the budget decreases by the number of defences made, but increases by 1 per timestep.
\par
Condition \textbf{4} counts the number of newly burnt vertices, ensuring that there only exists an entry with burnt vertex count $c$ if there is a corresponding strategy on which $c$ vertices burn by timestep $i$.
\par
We now show that computing these sets correctly answers the \textsc{Temporal Firefighter} problem. That is that entries exist in the sequence if and only if there is a corresponding strategy, and thus it is possible to check if $k$ vertices can be saved in temporal firefighter on the rooted temporal graph $((G, \lambda), r)$.
\begin{theorem}
Given a temporal graph $((G, \lambda), r)$ there is an entry $(D, B, g, c) \in L_i$ if and only if there exists a strategy for \textsc{Temporal Firefighter Reserve} on $((G, \lambda), r)$ such that $D$ and $B$ correspond to the vertices in $F_i$ that are defended and burnt respectively by timestep $i$, $g$ is the budget that will be available on timestep $i+1$, and $c$ is the total number of vertices that burn by timestep $i$.
\end{theorem}
\begin{proof}
We proceed by induction on $i$. After timestep $0$, only one vertex (the root) has burnt, and $L_0 = \{(\emptyset, \{r\}, 1, 1)\}$. Now suppose that the result holds at timestep $i-1$.
\par
We now show that if $(D', B', g-d+1, c') \in L_i$ then there is a corresponding strategy. For $(D', B', g-d+1, c')$ to be in $L_i$ there must be an entry $(D, B, g, c) \in L_{i-1}$ and set of vertices $A \subseteq V(E_i) \setminus (B \cup D)$ with $d = |A|$ such that conditions \textbf{1} through \textbf{4} hold. By our induction hypothesis there is a corresponding strategy $S_{i-1}$ for this entry such that $D$ and $B$ are the vertices in $F_{i-1}$ that are defended and burnt respectively by timestep $i-1$, $g$ is the budget available at the end of timestep $i-1$, and $c$ is the total number of vertices burnt by timestep $i-1$. If we take $S_{i-1}$ and extend it by defending the set of vertices $A$ on timestep $i$, then we obtain a strategy $S_i$ that we claim corresponds to $(D', B', g-d+1, c')$.

First see that by the definition of $A$, all the defences it contains are valid, as $A$ contains only vertices in $V(E_i) \subseteq F_i$ that have have not either already burnt or been defended.

The vertices that are newly defended on timestep $i$ in $S_i$ are only those in $A$. Thus the vertices that are defended in $F_i$ on timestep $i$ in $S_i$ are those that were already defended and are also in $F_i$, that is $D \cap F_i$, and those that are newly defended. Thus by condition \textbf{1} the vertices in $F_i$ that are defended by timestep $i$ in $S_i$ are those in $D'$.

The vertices that then burn on timestep $i$ in $S_i$ are all those temporally adjacent to the fire and not defended. Additionally, any vertex from which the fire spreads on timestep $i$ must be in $F_i$, as it must have an incident edge active at time $i$. For the same reason, any defended vertex that the fire would otherwise burn on timestep $i$ must also be in $F_i$. Thus, $N_i(B) \setminus D'$ is the set of vertices that newly burn. Therefore the vertices that have burnt in $F_i$ on timestep $i$ in $S_i$ are those that had already burnt and are also in $F_i$, that is $B \cap F_i$, and those that newly burn: $N_i(B) \setminus D'$. Thus by condition \textbf{2} the vertices in $F_i$ that have burnt by timestep $i$ in $S_i$ are those in $B'$.

The budget available on timestep $i+1$ in $S_i$ is the budget available on timestep $i$ incremented by 1, with $|A|$, the number of defences made on timestep $i$, subtracted. Thus by condition \textbf{3} the budget available at timestep $i+1$ in $S_i$ is $g-d+1$.

The set of vertices that newly burn after timestep $i$ in $S_i$ is all those temporally adjacent to the fire and not defended or already burning, so the number of such vertices is $|N_i(B)\setminus(B \cup D')|$. The total number of vertices to have burnt after timestep $i$ in $S_i$ is then $c + |N_i(B)\setminus(B \cup D')|$, thus by condition \textbf{4} the total number of vertices to have burnt is $c'$.
\par
We now show the converse: that if there is a strategy $S$ such that after time $i$ the sets of vertices that have been defended and burnt are $D_S$ and $B_S$ respectively, $g'$ is the available budget, and $c'$ is the total number of vertices to have burnt, then there is a corresponding entry $(D', B', g', c') \in L_i$, such that $D' = D_S \cap F_i$ and $B' = B_S \cap F_i$.

Consider the state at timestep $i-1$ if strategy $S$ is played. By our induction hypothesis there is a corresponding entry $(D, B, g, c) \in L_{i-1}$ where $D$ is the set of vertices in $F_{i-1}$ that are defended at time $i-1$, $B$ is the set of vertices in $F_{i-1}$ that are burnt at time $i-1$, $g$ is the budget that will be available at time $i$, and $c$ is the total number of vertices to have burnt at time $i-1$.

Let $A$ be the vertices defended at time $i$ in strategy $S$. As we consider only strategies that only defend vertices at time $i$ with incident edges at time $i$, and $A$ is a valid defence, we have that $A \subseteq V(E_i) \setminus (B \cup D)$.

By our induction hypothesis $D$ is the set of vertices in $F_{i-1}$ that are defended by time $i-1$, so the set of vertices in $F_i$ defended by time $i$ is $D_S \cap F_i = (D \cap F_i) \cup A$.

Again by the induction hypothesis $B$ is the set of vertices in $F_{i-1}$ that are burnt by time $i-1$. The only vertices from which the fire can spread on timestep $i$ are those that have an incident edge active at time $i$, and thus are in $F_i$. For the same reason the only defended vertices that would otherwise burn on timestep $i$ are in $F_i$. Therefore the vertices in $F_i$ burnt by time $i$ are $B_S \cap F_i = (B \cap F_i) \cup N_i(B) \setminus D_S = (B \cap F_i) \cup N_i(B) \setminus (D \cup A)$.

Finally, the budget available at time $i$ is $g$, and so the budget at time $i+1$ is $g' = g - |A| + 1$, and the number of vertices burnt after time $i-1$ is $c$, so the number of vertices burnt after time $i$ is $c' = c + |N_i(B) \setminus (B \cup D_S)| = c + |N_i(B) \setminus (B \cup (D \cap F_i) \cup A)|$.

Thus we see that, given $(D, B, g, c)$ is an entry in $L_{i-1}$, we have that $(D_S \cap F_i, B_S \cap F_i, g', c')$ satisfies conditions \textbf{1}-\textbf{4}, and thus is an entry in $L_i$.
\end{proof}

We now determine the runtime of computing all sets $L_i$, thus showing that \textsc{Temporal Firefighter} is FPT when parameterised by vertex-interval-membership-width.

\begin{theorem}
It is possible to solve \textsc{Temporal Firefighter} in time $O(8^\omega\omega\Lambda^3)$ for a rooted temporal graph $((G, \lambda), r)$ where $\Lambda$ is the lifetime of the graph, and $\omega$ is the vertex-interval-membership-width.
\end{theorem}
\begin{proof}
\textsc{Temporal Firefighter Reserve}, and therefore \textsc{Temporal Firefighter} can be answered by computing all sets $L_i$. Thus it suffices to show that each of these sets can be computed in the required time.

To compute $L_i$ every possible defence on a set of vertices with incident edges active at time $i$ must be considered for every entry in $L_{i-1}$.

First observe that the total number of burnt vertices on any given timestep $i$ is at most $\Sigma_{j=1}^i|V(E_j)| = O(\omega\Lambda)$, as on each timestep $j$ only vertices in $V(E_j)$ can burn, and $|V(E_j)| \leq 2|E_j| \leq 2\omega$, and for any timestep $i$ we have that $i \leq \Lambda$.

Now see that for any $i$, we have that $|L_i| = O(4^\omega\omega\Lambda^2)$ as $L_i \subseteq \mathcal{P}(F_i) \times \mathcal{P}(F_i) \times \{1, ..., \Lambda\} \times \{1, ..., \omega\Lambda\}$, and $|\mathcal{P}(F_i)| \times |\mathcal{P}(F_i)| = 2^{2\omega} = 4^\omega$.

Furthermore, on each timestep $i$ we only consider defending vertices in $V(E_i)$, and $|V(E_i)| \leq \omega$. Thus for each timestep there are at most $2^\omega$ defences to consider.

As described, for each timestep $i$ in the lifetime $\Lambda$ of the graph, it is necessary to compute every possible set of defences for every entry in $L_i$.  The overall complexity is therefore $O(4^\omega\omega\Lambda^2 \times 2^\omega \times \Lambda) = O(8^\omega\omega\Lambda^3)$ as required.
\end{proof}

\section{Conclusion}
In this paper we introduced the \textsc{Temporal Firefighter} problem, an extension of \textsc{Firefighter} to temporal graphs. We found that this problem is, like \textsc{Firefighter}, NP-Complete on arbitrary graphs, and in particular is NP-Complete on any underlying graph class for which \textsc{Firefighter} is NP-Complete. In order to try and identify places where \textsc{Temporal Firefighter} is tractable, we began by determining its complexity on several underlying graph classes for which it is known \textsc{Firefighter} can be solved in polynomial time.

Despite finding that \textsc{Temporal Firefighter} is NP-Complete on all but one underlying graph class that we considered, we were able to find a promising avenue for tractability in restricting the temporal structure of the graph, and found that the problem is FPT when parameterised by the temporal graph parameter vertex-interval-membership-width.

An interesting direction for future work on the problem would be to further investigate such temporal parameters. A natural first goal would be to determine the complexity of \textsc{Temporal Firefighter} when the maximum number of edges active each timestep is bounded. Furthermore, it would be worthwhile to determine the complexity of \textsc{Temporal Firefighter} when parameterised by interval membership width, a relative of vertex-interval-membership-width also introduced by Bumpus and Meeks \cite{bumpus2021edge}. This measure can be arbitrarily larger than the vertex-interval-membership-width.

\bibliography{bibliography}
\end{document}